\documentclass[10pt]{llncs}
\usepackage{amsmath,amssymb,array}
\usepackage{verbatim}
\usepackage{float}
\usepackage{tikz}

\usepackage{microtype}

\usepackage[algoruled,linesnumbered,noend]{algorithm2e}

\usepackage{comment}

\mathchardef\mhyphen="2D

\newcommand{\LongPath}{{\sf $h$-Path}}

\newcommand{\SetCover}{{\sf Minimum Set Cover}}

\newcommand{\SeqComp}{{\sf  String to Graph 
Compatibility Matching}}

\newcommand{\SeqAlign}{{\sf String to Graph Approximate Matching}}

\newcommand{\SeqAlG}{{\sf String to Graph 
Restricted  Approximate Matching}}

\begin{document}

\title{Complexity Issues of String to Graph
Approximate Matching\thanks{Riccardo Dondi dedicates
the paper to the memory of his beloved father, Gilberto, who passed away on November 26, 2019.}}
\author{Riccardo Dondi\inst{1}
\and Giancarlo Mauri \inst{2}
\and Italo Zoppis \inst{2}}
\institute{
Universit\`a degli Studi di Bergamo, Bergamo, Italy\\
\and
Universit\`a degli Studi di Milano-Bicocca, Milano - Italy\\
\email{riccardo.dondi@unibg.it,mauri@disco.unimib.it, zoppis@disco.unimib.it}
}

\maketitle

\begin{abstract}
The problem of matching a query string to a
directed graph, whose vertices are labeled
by strings, has application in different
fields, from data mining to computational biology.
Several variants of the problem have been
considered, depending on the fact that
the match is exact or approximate and,
in this latter case, which edit operations
are considered and 
where are allowed.
In this paper we present results on the 
complexity of the approximate matching
problem, where edit operations 
are symbol substitutions and are
allowed only on
the graph labels or
both on the graph labels and the query string.
We introduce a variant of the problem
that asks
whether there exists a path
in a graph that represents a query string 
with any number of edit operations
and we show that is is NP-complete,
even when labels have length one and in the
case the alphabet is binary.
Moreover, when it is parameterized by 
the length of the input string and
graph labels have length one, 
we show that the problem is  fixed-parameter 
tractable and it is unlikely to admit a 
polynomial kernel. 
The NP-completeness of this problem leads to the
inapproximability 
(within any factor) 
of the approximate matching 
when edit operations are allowed only on the
graph labels.
Moreover, we show that the variants of
 approximate string 
matching to graph we consider are 
not fixed-parameter tractable, when the parameter
is the number of edit operations, even
for graphs that have distance one from a DAG.
The reduction for this latter result 
allows us to prove the inapproximability 
of the variant
where edit operations can be applied both on the query string
and on graph labels. 
\end{abstract}

\section{Introduction}
\label{sec:intro}

Given a query string $s$
and a directed graph $G$ whose
vertices are labeled with strings (referred
as labeled graph), the matching and the approximate matching of $s$ to $G$
ask for a path (not necessarily simple)
in $G$ that represents $s$, that is
by concatenating the labels of the vertices
on the path we obtain $s$ or an approximate occurrence of $s$.

The matching and the approximate matching of a query
string to a labeled graph have applications
in different areas, from graph databases and
data mining to genome
research. The problems have been introduced
in the context of  pattern matching
in hypertext \cite{DBLP:conf/cpm/Akutsu93,DBLP:journals/jal/AmirLL00,DBLP:conf/cpm/ParkK95,Manber1992}, but
have found recently new applications.
Indeed in computational biology 
a representation of variants of related sequences
is often provided by a labeled graph \cite{Pevzner2001,DBLP:conf/eccb/Myers05}
and the query of a string
in a labeled graph has found application in
computational pan-genomics \cite{DBLP:journals/jcb/NguyenHZREAKHP15,pangenomics}.

The exact matching problem is known 
to be in P \cite{DBLP:conf/cpm/Akutsu93,DBLP:journals/jal/AmirLL00,DBLP:conf/cpm/ParkK95}.
Furthermore, conditional lower bounds for this problem has been recently given in \cite{DBLP:conf/icalp/EquiGMT19}.

The approximate string to graph matching problem, referred
to \SeqAlign{},
has the goal of minimizing the number of edit operations 
(of the query string
or of the labels of the graph) such that there 
exists a path $p$ in $G$
whose labels match the query string.
\SeqAlG{} denotes the variant where
edit operations are allowed only on the graph labels.
\SeqAlign{} and \SeqAlG{} are
known to be NP-hard \cite{DBLP:journals/tcs/Navarro00}, 
even for binary alphabet \cite{DBLP:conf/recomb/JainZGA19}.
When the edit operations are allowed only on the query 
string, then 
\SeqAlign{} is polynomial-time 
solvable \cite{DBLP:conf/recomb/JainZGA19}. 
Moreover, when the input graph is a Directed Acyclic
Graph (DAG), \SeqAlign{} and \SeqAlG{} are polynomial-time 
solvable \cite{Manber1992}.

In this contribution, we consider  the \SeqAlign{} problem
and the \SeqAlG{} problem,
with the goal of deepening the understanding
of their complexity.
Notice that the edit operations we consider are
symbol substitutions 
of the graph labels or of the query string.
Other variants with different
edit operations have been considered in literature \cite{DBLP:journals/jal/AmirLL00,DBLP:conf/recomb/JainZGA19}.

We introduce a variant of \SeqAlG{}, called
\SeqComp{}, that asks whether it is possible to 
find an occurrence of a query string
in a graph with any number of edit operations of the graph labels.
This decision problem  is helpful to characterize
whether a feasible solution of
\SeqAlG{}
exists or not.
We show in Section \ref{sec:hardnessComp} that \SeqComp{} 
is NP-complete,
even when the labels of the graph have length one or
when the alphabet is binary. 
The reduction shows also that 
\SeqComp{} when parameterized by the length
of the query string 
is unlikely to have a polynomial kernel\footnote{
A problem parameterized by parameter $t$ admits a polynomial kernel 
if there exists a polynomial-time
algorithm that reduces the instance of the problem so that it has 
a size which is a polynomial in $t$.} (for details
on kernelization we refer to \cite{Niedermeier:2006,DBLP:series/txcs/DowneyF13}).
A consequence of the intractability of \SeqComp{} 
is that 
\SeqAlG{} cannot be approximated within any factor in polynomial time.
Notice that if we allow edit operations of the query string,
then the existence of a path that represents
an approximate matching of the query string can be decided 
in polynomial time.
Indeed, it is enough to check whether 
the input graph
contains a (non necessarily simple) path $p$ in $G$ that represents
a string of length $|s|$. 

We consider in Section \ref{sec:hardAlign} the 
parameterized complexity of 
\SeqAlG{} and of \SeqAlign{} and
we show that they are W[2]-hard 
when parameterized by the number of edit operations,
even for a labeled graph having distance one from a 
DAG.
This result shows that, while \SeqAlG{} and \SeqAlign{} 
are solvable
in polynomial time when the labeled graph is a 
DAG \cite{Manber1992}, even for graphs that are very close to DAG
they become hard. 
The reduction designed to prove this latter result 
allows us to show 
that \SeqAlign{} is not approximable within factor 
$\Omega(\log (|V|))$ and $\Omega(\log (|s|))$, for a labeled graph $G=(V,E)$ and a query string $s$.

In Section \ref{sec:paracompl1}, we provide
a fixed-parameter tractable algorithm for \SeqComp{},
when parameterized by size of the query string and when the graph labels have
length one.
We conclude the paper in Section \ref{sec:conclusion} with some
open problems, while in Section \ref{sec:def} we
introduce some definitions
and the problems we are interested in.
Some of the proofs are not included due to page limit.

\section{Definitions}
\label{sec:def}
Given an alphabet $\Sigma$ and a string 
$s$ over $\Sigma$, we denote by $|s|$ the length
of $s$, by $s[i]$, with $1 \leq i \leq |s|$,
the $i$-th symbol of $s$ and
by $s[i,j]$, with $1 \leq i \leq j \leq |s|$,
the substring of $s$ that starts at position $i$
and ends at position $j$. 

Every graph we consider in this paper
is \emph{directed}.
Given a graph $G=(V,E)$ and a vertex $v \in V$, 
we define
$N^+(v)= \{ u \in V: (v,u) \in E \}$
and $N^-(v)= \{ w \in V: (w,v) \in E \}$.

A labeled graph $G = (V,E,\sigma)$ is a graph whose
vertices are labeled with strings, formally
assigned by a labeled function
$\sigma: V \rightarrow \Sigma^*$, 
where $\Sigma$ 
is an alphabet of symbols.
Notice that $\sigma(v)$, with $v \in V$, denotes
the string 
associated by $\sigma$ to vertex $v$.
Let $p = v_1 v_2 \dots v_z$ be a path (non necessarily
simple) 
in $G$, the set of vertices that induces $p$
is denoted by $V(p)$ and the string 
associated with $p$ is defined as
$\sigma(p) = \sigma(v_1) \sigma(v_2) \dots \sigma(v_z)$,
that is $\sigma(p)$ is obtained by concatenating
the strings that label the vertices of path $p$.

Consider a string $s$ on alphabet $\Sigma$ and a
labeled graph $G=(V,E,\sigma)$.
We say that a path $p$ in $G$
is an occurrence of $s$ if $\sigma(p) = s$; 
in this case we call $\sigma(p)$ 
an exact matching of $s$ and we say that $p$ matches $s$.

An edit operation of a string $s$ is a substitution 
of the symbol in a position $i$, 
with $1 \leq i \leq |s|$, of $s$ with a different symbol
in $\Sigma$.
An edit operation of $G=(V,E, \sigma)$ is an edit operation of a string
$\sigma(v)$, with $v \in V$.
A path $p$ in $G$ is an \emph{approximate matching}
of $s$ if, after $k_1 \geq 0$ edit operations
of labels of $G$, $\sigma(p) = s'$,
where $s'$ is a string obtained with $k_2 \geq 0$ 
edit operations of $s$.
In this case, we say that
the approximate matching requires
$k = k_1 + k_2$ edit operations.
We say that $p$ in $G$ is a 
\emph{restricted approximate matching}
of $s$, if, after after $k \geq 0$ edit operations
to labels of $G$, $s = \sigma(p)$ (that is the edit operations are allowed only on the labels of $G$).

Consider a path $p$ that matches 
(exactly, approximately
or restricted approximately)
the query string $s$. 
If 
position $i$, $1 \leq i \leq |s|$, in $s$ 
and 
the $j$-th position, 
$1 \leq j \leq |\sigma(u)|$, of the label 
of vertex $u$ in $p$ match (possibly after an edit operation),
we say that position $i$ is mapped in $\sigma(u)[j]$;
if $|\sigma(u)|=1$, by slightly abusing the notation,
we say that position $i$ is mapped in $u$.

Next, we define the first combinatorial problem
we are interested in.

\begin{problem}
\label{Problem:DefAlign1} \SeqAlign{}\\
\textbf{Input}: A labeled graph $G=(V,E, \sigma)$ and a 
query string $s$, both on alphabet $\Sigma$.\\
\textbf{Output}: An approximate matching $p$
of $s$ that requires the minimum number
of edit operations.
\end{problem}
We define now
the variant of the problem, 
called \SeqAlG{}, where edit operations are allowed
only on the labels of the labeled graph.

\begin{problem}
\label{Problem:DefAlign2} \SeqAlG{}\\
\textbf{Input}: A labeled graph $G=(V,E, \sigma)$ and a 
query string $s$, both on alphabet $\Sigma$.\\
\textbf{Output}: A restricted approximate matching $p$
of $s$ that requires the minimum number of edit operations.
\end{problem}

Consider a labeled graph $G=(V,E,\sigma)$ and a 
query string $s$ over $\Sigma$. 
If there exists a path $p$ in $G$
which is a restricted approximate matching
of $s$, we say that $p$ is \emph{compatible}
with $s$.
Notice that the definition of
compatibility does not put any bound
on the number of edit operations of graph
labels and that
no edit operation is allowed on the query
string. 
In this paper, we introduce a decision problem, called
\SeqComp{}, related to \SeqAlG{}, that asks 
whether there exists a path in
$G=(V,E,\sigma)$ compatible with $s$.

\begin{problem}
\label{Problem:DefComp} \SeqComp{}\\
\textbf{Input}: A labeled graph $G=(V,E, \sigma)$, a query string
$s$, both on alphabet $\Sigma$.\\
\textbf{Output}: Does there exist a path in $G$
that is compatible with $s$?
\end{problem}

\section{Hardness of \SeqComp{}}
\label{sec:hardnessComp}

In this section we consider the computational complexity of \SeqComp{} and we prove that the
problem is indeed NP-complete and it is unlikely
to admit a polynomial kernel.
This result, as discussed in Theorem \ref{teo:approx},
is not only interesting to characterize the
complexity of \SeqComp{}, but also
to give insights into the approximation complexity
of \SeqAlG{}.

We start by proving that \SeqComp{} 
is NP-complete when the labels of the graph have length one,
via a reduction from 
the \LongPath{} problem. The reduction is inspired by that in \cite{DBLP:journals/jal/AmirLL00} to prove the
NP-hardness of \SeqAlG{}.
Then we modify the reduction so that it holds also for 
binary alphabet.
We 
recall the definition of \LongPath{}, 
which is known to be NP-complete \cite{garey}.

\begin{problem}
\label{Problem:Def2} \LongPath{}\\
\textbf{Input}: A directed graph $G=(V_L,E_L)$.\\
\textbf{Output}: Does there exist a simple path 
in $G_L$ of length $h$?
\end{problem}

\subsection{Graph Labels of Length One}
\label{subsec:unbound}

Consider a graph $G_L=(V_L,E_L)$, 
with $V_L=\{ v_1^l, \dots, v_n^l \}$,
which is an instance of \LongPath{},
we define an instance of \SeqComp{} consisting 
of a labeled graph $G=(V,E, \sigma)$ and 
a query string $s$.

First, define the alphabet $\Sigma$ as follows:
$
\Sigma = \{ x_i : 1 \leq i \leq n  \} \cup 
\{ y_i : 1 \leq i \leq h \}.$

The labeled graph $G=(V,E, \sigma)$ is defined as follows: 
\[
V = \{ v_i : v_i^l \in V, 1 \leq i \leq n \},
 \hspace{.5cm }
E = \{ (v_i,v_j) : (v_i^l, v_j^l) \in E_L \}.
\]

The labelling function $\sigma: V \rightarrow \Sigma^*$ of the graph vertices is defined as follows: 
$
\sigma(v_i) = x_i \text{, for each $i$ with $1 \leq i \leq n$}.$

Finally, we define the query string $s = y_1 y_2 \dots y_h$. 

The following lemma allows us to
prove the hardness of \SeqComp{}.

\begin{lemma}
\label{lem:GeneralComp}
Let $G_L=(V_L, E_L)$ be a graph instance of \LongPath{}
and let $(G=(V,E,\sigma),s)$ be the corresponding instance
of \SeqComp{}. There exists a simple path of length $h$ in
$G_L$ if and only if there exists a path in $G$
compatible with $s$.
\end{lemma}
\begin{proof}
Consider a simple path 
$v_{i_1}^l v_{i_2}^l \dots  v_{i_h}^l$
in $G_L$. Then consider the corresponding path 
$v_{i_1} v_{i_2} \dots  v_{i_h}$
in $G$ and edit the symbol of each vertex
$v_{i_j}$, with $1 \leq j \leq h$, so that it 
is associated with symbol $y_i$. It follows
that $p$ matches $s$.
Then $v_{i_1} v_{i_2} \dots  v_{i_h}$ is 
a path of $G$ compatible with 
$s$.

Consider a path $p= v_{i_1} v_{i_2} \dots  v_{i_h}$
in $G_L$ compatible with $s$. 
Notice that $p$ must be a simple path,
since $s$ consists of $h$ distinct symbols.
As a consequence, the corresponding path 
$v_{i_1}^l v_{i_2}^l \dots  v_{i_h}^l$
in $G_L$
is a simple path of length $h$.
\qed
\end{proof}
Lemma \ref{lem:GeneralComp} 
and the NP-completeness of \LongPath{} \cite{garey}
allow to prove the following result.

\begin{theorem}
\label{teo_hard1}
\SeqComp{} is NP-complete even when the labels of the graph have length one.
\end{theorem}
\begin{proof}
Consider a graph $G=(V,E,\sigma)$, and a
path $p$ in $G$. 
It can be checked in
polynomial time if $p$ is compatible with $s$, 
so \SeqComp{} is in NP.

The reduction constructs a labeled graph with labels
of length one. From Lemma \ref{lem:GeneralComp} and from 
the hardness of \LongPath{} \cite{garey}, 
it follows that \SeqComp{} is NP-hard even when the labels
of the graph have length one.
\qed
\end{proof}


Notice that the reduction we have described 
is also a Polynomial Parameter Transformation \cite{DBLP:journals/tcs/BodlaenderJK13} from
\LongPath{} parameterized by $h$ to \SeqComp{} parameterized
by $|s|$,
as $|s|=h$.
Since \LongPath{} when parameterized by $h$ does not
admit a polynomial kernel unless $NP \subseteq coNP/Poly$~\cite{DBLP:journals/jcss/BodlaenderDFH09}, the reduction
leads to the following result.

\begin{corollary}
\label{cor:kernel}
The \SeqComp{} problem parameterized
by $|s|$ does not admit a polynomial kernel unless 
$NP \subseteq coNP/Poly$ even when the labels of the graph have length one.
\end{corollary}

\subsection{Binary Alphabet}

Next, we show that the \SeqComp{} problem is
NP-complete even on binary alphabet.
The reduction is similar to the reduction of 
the Section \ref{subsec:unbound}, except for
the definition of 
the query string $s$ and the labeling 
$\sigma : V \rightarrow \Sigma^*$ of the 
labeled graph.

Consider a graph $G_L=(V_L,E_L)$, 
with $V_L=\{ v_1^l, \dots, v_n^l \}$, that is an instance of \LongPath{},
we define a corresponding instance $(G=(V,E, \sigma),s)$ of \SeqComp{}.
The alphabet is binary, 
hence $\Sigma = \{ 0,1 \}$.
Next, we define the labeled graph $G=(V,E, \sigma)$. The sets $V$ of vertices and $E$ of edges are defined
as in Section \ref{subsec:unbound}.
For each $v_i \in V$, with $1 \leq i \leq h$,
$\sigma(v_i) = 0^{h}$, namely it is a string
consisting of $h$ occurrences of symbol $0$.

The construction of the query string $s$ requires the
introduction of strings $s_i$, 
with $1 \leq i \leq h$, having length $h$ and defined as follows:

\[
s_i[i] = 1; \hspace{1cm}  s_i[j] = 0, \text{ with $1 \leq j \leq h$ and $j \neq i$.}
\]

%
%


Finally, $s$ is defined as the concatenation 
of $s_1$,$s_2$, $\dots s_n$, that is
$s = s_1\ s_2\ \dots s_n$.



Next, we prove the correctness of
the reduction.

\begin{lemma}
\label{lem:BinaryComp}
Let $G_L=(V_L, E_L)$ be a graph instance of \LongPath{}
and let $(G=(V,E,\sigma),s)$ be the corresponding instance
of \SeqComp{} on binary alphabet. 
There exists a simple path of length $h$ in
$G_L$ if and only if there exists a path 
compatible with $s$ in $G$.
\end{lemma}
\begin{proof}
Consider a simple path $v_{i_1}^l v_{i_2}^l \dots  v_{i_h}^l$
in $G_L$. Then consider the corresponding path 
$v_{i_1} v_{i_2} \dots  v_{i_h}$
in $G$ and edit the label of each vertex $v_{i_j}$,
with $1 \leq j \leq h$, such that is associated
with string $s_j$.
Then the resulting string is an exact match of $s$, hence $v_{i_1} v_{i_2} \dots  v_{i_h}$ is a
path compatible with $s$.

Consider a path $p = v_{i_1} v_{i_2} \dots  v_{i_h}$
in $G$ that is compatible with $s$. 
Since $\sigma(p)$  must match $s$ after some symbol substitutions and, by construction,
$|\sigma(v_j)| = |s_l|$,
for each $1 \leq j \leq n$ and $1 \leq l \leq h$,
it follows that the positions of $s_l$, $1 \leq l \leq h$,
are mapped to the positions of $\sigma(v_{i_t})$, 
for some $t$ with $1 \leq t \leq h$.
Moreover, since $s_l \neq s_q$, with $t \neq q$,
all the vertices in $p$ are distinct
and $p$ is a simple path in $G$ of length $h$.
As a consequence the corresponding path 
$v_{i_1}^l v_{i_2}^l \dots  v_{i_h}^l$ in $G_L$
is a simple path of length $h$, thus concluding the proof.
\qed
\end{proof}

Thus, based on Lemma \ref{lem:BinaryComp}, we can prove the following result.

\begin{theorem}
\label{teo_hard2}
\SeqComp{} is NP-complete even on binary alphabet.
\end{theorem}
\begin{proof}
As for Theorem \ref{teo_hard1}, 
given a path $p$ in $G$, it can be checked in
polynomial time if $p$ is compatible with $s$, 
so \SeqComp{} is in NP.
The reduction defines $\Sigma=\{0,1\}$, hence a binary
alphabet. From Lemma \ref{lem:BinaryComp} and from 
the hardness of \LongPath{} \cite{garey}
it follows that
\SeqComp{} on binary alphabet is NP-hard.
\qed
\end{proof}

The results of Theorem \ref{teo_hard1} and
Theorem \ref{teo_hard2} have a consequence
not only on the complexity of \SeqComp{}, 
but also on the approximation of \SeqAlG{}.

\begin{theorem}
\label{teo:approx}
The \SeqAlG{} problem cannot be approximated within any factor in polynomial time, unless P = NP,
even when the labels of the graph have length one
or when the alphabet is binary.
\end{theorem}
\begin{proof}
The NP-completeness of \SeqComp{} implies that,
given an instance $(G=(V,E,\sigma), s)$,
even deciding whether there 
exists a feasible solution of \SeqAlG{},
with any number of edit operations in $G$,
is NP-complete. Hence if there exists a polynomial-time
approximation algorithm $\mathcal{A}$
for \SeqAlG{}, with some approximation
factor $\alpha$, 
it follows that $\mathcal{A}$ can be
used to decide the \SeqComp{} problem:
if $\mathcal{A}$ returns an approximated 
solution for \SeqAlG{}
with input $(G,s)$, then it follows that
there exists a path in $G$ compatible with $s$,
if $\mathcal{A}$ does not return an approximated solution for \SeqAlG{}
with input $(G,s)$, then
there is no path in $G$ compatible with $s$.
Since \SeqComp{} is NP-complete,
when the labels of the graph have length one (by Theorem \ref{teo_hard1} )
and on binary alphabet (by Theorem \ref{teo_hard2}), 
then there does not
exist a polynomial-time approximation algorithm
with any approximation factor for \SeqAlG{} when the graph
labels have length one
or when the alphabet is binary, unless P = NP.
\qed
\end{proof}

\section{Hardness of Parameterization} 
\label{sec:hardAlign}

In this section, we consider the parameterized complexity
of \SeqAlG{} and \SeqAlign{}.
The reduction we present allows us to prove
that \SeqAlG{} and \SeqAlign{},
when parameterized by the number of edit operations,
are W[2]-hard
for a labeled graph having distance one from a DAG.
Moreover, the same reduction will allow us to
prove that \SeqAlign{} 
is not approximable within factor 
$\Omega(\log (|V|))$ and $\Omega(\log (|s|))$.

%

We prove these results by presenting a reduction, that is
parameterized \cite{Niedermeier:2006,DBLP:series/txcs/DowneyF13} and approximate preserving \cite{DBLP:books/daglib/0030297}, from  the \SetCover{} problem.  We recall here the definition
of \SetCover{}.

\begin{problem}
\label{Problem:Def3} \SetCover{}\\
\textbf{Input}: A collection $C=\{S_1, \dots, S_m\}$ of sets over a universe $U=\{ u_1, \dots, u_n\}$.\\
\textbf{Output}: A subcollection $C'$ of $C$ of minimum 
cardinality such that
for each $u_i \in U$, with $1 \leq i \leq n$, 
there exists a set in $C'$ containing
$u_i$.
\end{problem}

First, we focus on \SeqAlG{}, then we show that
the same reduction can be applied to \SeqAlign{}.

Given an instance $(U,C)$ of \SetCover{},
in the following we define 
an instance $(G=(V,E, \sigma),s)$
of \SeqAlG{} (see Fig. \ref{fig:Whardness} for an example). 
We start by defining the alphabet $\Sigma$:
\[
\Sigma = \{ x_i: 0 \leq i \leq m \} \cup \{ y_i: 1 \leq i \leq n \} \cup \{ z \}.
\]
Then, we define the labeled graph $G=(V,E, \sigma)$:
\[
V =  \{v_i : 0 \leq i \leq m\} \cup 
\{ v_{i,j}: 1 \leq i \leq m, 1 \leq j \leq |S_i| \} 
\]
\[
E =  \{(v_0 ,v_i) : 1 \leq i \leq m \} \cup 
\{ (v_i, v_{i,j}): 1 \leq i \leq m, 1 \leq j \leq  |S_i| \} \cup \]
\[
\{ (v_{i,j},v_0): 1 \leq i \leq m, 1 \leq j \leq |S_i| \}.
\]
Now, we define the labeling $\sigma$ 
of the vertices of $G$:
\begin{itemize}


\item  $\sigma(v_i)= x_i$, $0 \leq i \leq m$

\item  $\sigma(v_{i,l})= y_j$, $1 \leq i \leq m$, $1 \leq l \leq |S_i|$ and $1 \leq j \leq n$,
where the $l$-th element of $S_i$ is $u_j$ 
(based on some ordering of the elements in $S_i$)

\end{itemize}

The query string $s$ is defined as follows:
$s = x_0\ z\ y_1\ x_0\ z\ y_2 \dots  x_0\ z\ y_n$.

\begin{figure}
\centering
\includegraphics[scale=.35]{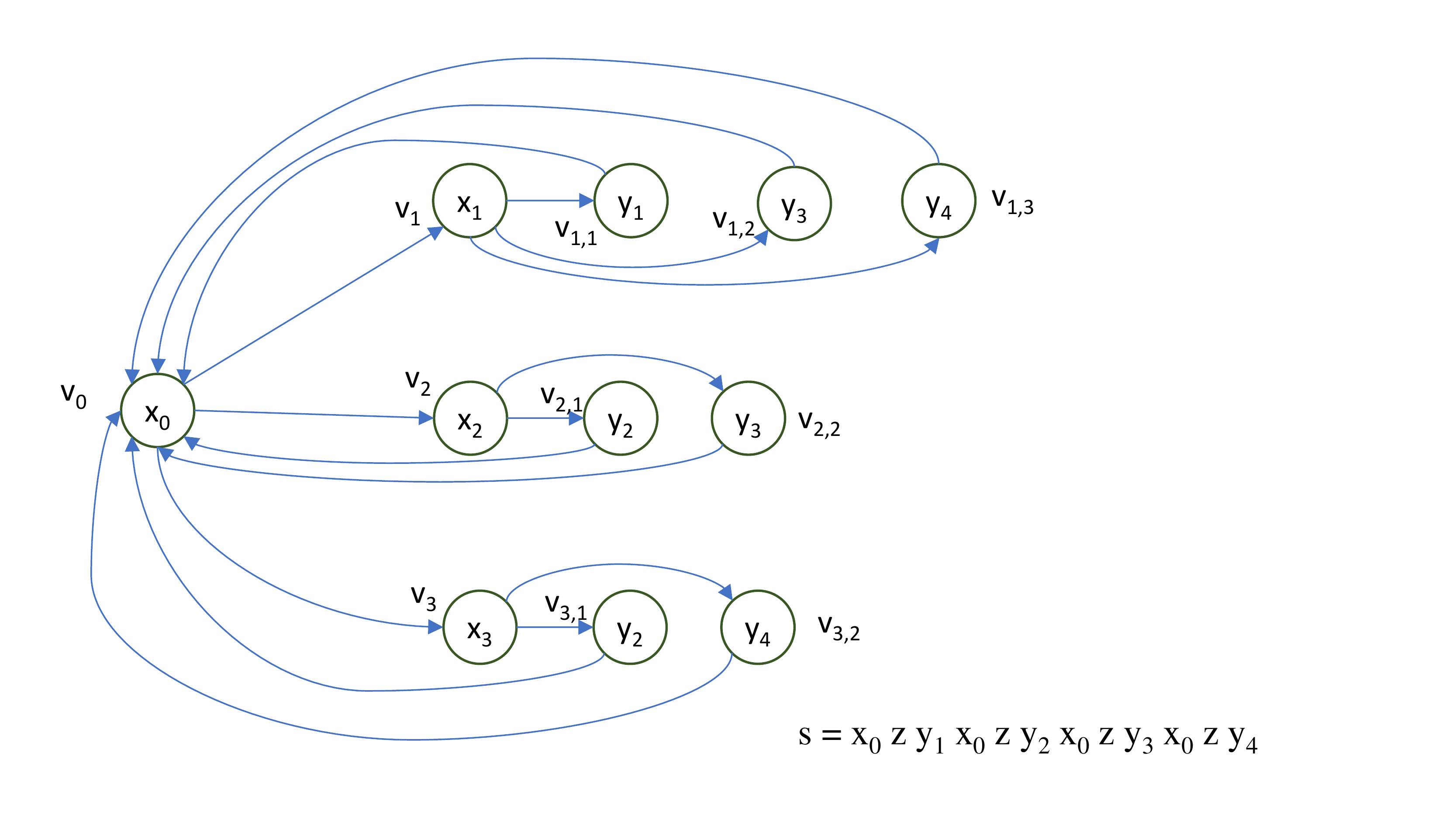}
\caption{A labeled graph $G$ and a query string $s$ associated with the following instance
of \SetCover{}:  $U=\{u_1, u_2, u_3, u_4\}$; 
$S_1=\{ u_1, u_3, u_4\}$, $S_2 = \{ u_2, u_3\}$,
$S_3 = \{u_2, u_4\}$. Inside each vertex we represent
its label.
}
\label{fig:Whardness}
\end{figure}

First, we prove that the labeled graph $G$,
has distance one from a DAG, that is by removing
a vertex of $G$ (namely, $v_0)$, we obtain a DAG. 

\begin{lemma}
\label{lem:pathwidth}
Let $(C,U)$ be an  instance of \SetCover{}
and let $(G=(V,E,\sigma),s)$ be the corresponding instance
of \SeqAlG{}.
Then, $G$ has distance one from a DAG.
\end{lemma}
\begin{proof}
Consider graph $G=(V,E, \sigma)$. By removing
vertex $v_0$, we obtain a set of $m$ disconnected 
DAGs, each one induced by vertices $v_i$ and $v_{i,j}$, 
with $1 \leq i \leq m$ and $1 \leq j \leq |S_i|$.
Hence the graph we obtain from the removal of $v_0$ 
is a DAG and the lemma holds.
\qed
\end{proof}

Next, we present the main result to prove the correctness of the reduction.

\begin{lemma}
\label{lem:sub1}
Let $(C,U)$ be an  instance of \SetCover{}
and let $(G=(V,E,\sigma),s)$ be the corresponding instance
of \SeqAlG{}.
There exists a cover $C'$ of $U$ of cardinality $h < n$ 
if and only if there exists a solution of \SeqAlG{}
that requires $h$ edit operations.
\end{lemma}
\begin{proof}
Consider 
$C'= \{S_{i_1}, S_{i_2}, \dots , S_{i_h}\}$, with
$C' \subseteq C$, that covers $U$. 
Starting from $C'$, define a set $C^*$ consisting of sets $S^*_{i_j}$, $1 \leq j \leq h$, defined as follows:
$S^*_{i_j}$ contains the elements of $U$ covered
by $S_{i_j}$ and not covered by a set
$S_{i_l}$, with $l < j$, in $C'$. Hence, since $C'$
covers $U$, it follows
that by construction each element in $U$ is covered by 
exactly one set in $C^*$.

For each element $u_i \in U$, with $1 \leq i \leq n$,
we denote by $S^*(u_i)$ the set of $C^*$ that covers $u_i$,
and by $v(S^*(u_i))$ the corresponding vertex 
in $\{v_1, \dots , v_n\} \subseteq V$. 
Moreover, given $v_j = S^*(u_i)$,
we denote by $y(S^*(u_i))$ the vertex
$v_{j,l}$, with $1 \leq l \leq |S_j|$, such that
$\sigma(v_{j,l}) = y_i$, with $1 \leq l \leq |S_j|$.

Consider the following path $p$ in $G$:
\[
p =  v_0,\ v(S^*(u_1)),\ y(S^*(u_1)),\ v_0,\ v(S^*(u_2)),\ v_0 ,\dots  ,\ v(S^*(u_n)),
\ y(S^*(u_n)).
\]

Now, the symbol associated with each $v(S^*((u_i))$, 
with $1 \leq i \leq n$, is edited to symbol $y$.
The string we obtain is exactly $s$,
since $\sigma(v_0)=x_0$, each symbol associated with 
$v(S^*(u_i))$ is edited to  $y$, 
$1 \leq i \leq n$, and $\sigma(y(S^*(u_i))) = y_i$,
$1 \leq i \leq n$. Moreover, notice that $p$ is
obtained by editing $h$ labels of vertices of $G$. 

We present the second direction of the proof.
Consider a path $p$ in $G$ such that  
$p$ is a restricted approximate matching of $s$ requiring 
at most $h$ edit operations of the labels
of vertices in $p$.
First, we prove some properties of $G$. 
If $v_0$ is removed from $G$, then 
the resulting graph $G'$ contains paths 
consisting of at most $2$ vertices.
Since $|s|= 3n$, there is no path in $G'$ 
that can be a restricted approximate matching of $s$. 
This implies that at least one
position of $s$ is mapped in $v_0$.

Now, assume that the first vertex of $p$ is not $v_0$.
Assume that the first position of $s$ is
mapped in $v_i$, for some $i$ with $1 \leq i \leq m$.
By construction, 
$p = v_i\ v_{i,j}\ v_0\ v_l\ v_{l,t}\ v_0 \dots$,
since $N^+(v_i) = \{ v_{i,j}: 1 \leq j \leq |S_j| \}$,
$N^+(v_{i,j}) = \{ v_0 \}$ and 
$N^+(v_0) = \{ v_i: 1 \leq i \leq m \}$.
Then each occurrence of  
a symbol $y_q$, $1 \leq q \leq n$,
in $s$ is mapped
in $v_0$, while the symbol associated with $v_0$ 
can be at most one of
$y_1 ,\dots , y_n$, thus there is no path 
in $G$ that starts with a vertex $v_i$ and that is
a restricted approximate matching of $s$.

Assume that the first vertex of $p$ is
some vertex $v_{i,j}$, with $1 \leq i \leq m$ and 
$1 \leq j \leq |S_i|$.
By construction, 
$p = v_{i,j}\ v_0\ v_l\ v_{l,t}\ v_0 \dots$.
Hence each position of $s$ containing $z$ is mapped
in vertex $v_0$, while each position of $s$
containing $y_t$, $1 \leq t \leq n$, is mapped
in a vertex $v_q$, with $1 \leq q \leq m$. This last
mapping requires $n > h$ edit operations
of labels of vertices of $G$, violating
the hypothesis that at most $h < n$ edit operations
are applied. 

We can conclude that if $p$ is
a restricted approximate matching of $s$ requiring
$h < n$ edit operations, then 
$v_0$ must be the first vertex
of $p$. 
It follows that each label of a vertex $v_i$, 
$1 \leq i \leq n$, in path $p$ 
must be edited to $z$.
Consider the case that  position $t$ of $s$,
$1 \leq t \leq |s|$,
where $s[t] = y_q$, $1 \leq q \leq n$, 
is mapped to some vertex $v_{i,j}$, 
with $1 \leq i \leq m$ and $1 \leq j \leq |S_i|$,
such that $\sigma(v_{i,j}) \neq y_q$, and that hence
the label of $v_{i,j}$ is edited to $y_q$.
Let $v_a$, with $1 \leq a \leq m$, be the vertex
that precedes $v_{i,j}$ in $p$.
Then, we can modify $p$, so that the number
of edit operations are not increased,
by replacing 
$v_a$ with a vertex $v_b$, with $1 \leq b \leq m$,
and $v_{i,j}$ with $v_{b,l}$, with $1 \leq l \leq |S_b|$,
so that $\sigma(v_{b,l}) = y_q$,
and by editing the label of 
$v_b$ (if it is no already edited) to $z$.
This implies that  
the only vertices of $p$ whose labels are edited
are vertices $v_i$, 
$1 \leq i \leq m$. 

Now, we can define a solution $C'$ of \SetCover{} consisting
of $h$ sets as follows:
$
C' = \{ S_i: \text{the label of vertex $v_i$ in $p$ is edited to } z , 1 \leq i \leq m \}.
$
Since at most $h$ labels of vertices of $p$ are edited (to $z$),
it follows that at most $h$ sets belong to $C'$.
Furthermore, since each vertex with label $y_j$,
$1 \leq j \leq n$, is connected to a vertex
$v_i$ in $p$, $1 \leq i \leq m$, by construction
it follows that 
each element of $U$ belongs to some set in $C'$. 
\qed
\end{proof}

Based on Lemma \ref{lem:pathwidth} and on 
Lemma \ref{lem:sub1}, we can prove the 
following result. 

\begin{theorem}
\label{teo:W1Hard}
The \SeqAlG{} problem is W[2]-hard when parameterized by the number of edit operations, even
when the input graph has distance one from a DAG.
\end{theorem}
\begin{proof}
Notice that, by Lemma \ref{lem:pathwidth}, $G$ has
distance one from a DAG.
The W[2]-hardness of \SeqAlign{} follows from
Lemma \ref{lem:sub1} and from the W[2]-hardness
of \SetCover{} \cite{DBLP:journals/tcs/PazM81}.
\qed
\end{proof}

Next, we show that the same reduction 
allows us to prove the W[2]-hardness and the
inapproximability of \SeqAlign{}.
Essentially, we will prove that
we can avoid edit operations of the 
query string.

\begin{theorem}
\label{teo:HardBothSubProbl}
The \SeqAlign{} problem is W[2]-hard when parameterized
by the number of edit operations, even
when the input graph has distance one from a DAG.
Moreover, \SeqAlign{} cannot be approximated within
factor $\Omega(\log (|V|))$ and $\Omega(\log (|s|))$,
unless $P=NP$, even
when the input graph has distance one from a DAG.
\end{theorem}
\begin{proof}
We consider the same construction given in the reduction 
from \SetCover{} to \SeqAlG{}.
Notice that, if there exists a cover $C'$ of $U$ of 
cardinality $h < n$, we can compute
in polynomial time a solution of \SeqAlign{} 
with at most $h$ edit operations as in Lemma \ref{lem:sub1}.
Now, consider a solution of
\SeqAlign{} on instance $G=(V,E,\sigma)$
that requires at most $h < n$ edit operations, we show that
we can compute in polynomial time
a cover $C'$ of $U$ of cardinality $h$.
We prove that, given 
a solution of \SeqAlign{} on instance $G=(V,E,\sigma)$,
we can restrict ourselves 
to edit operations on graph labels.

Let $p$ be a path in $G$ that approximately matches $s$.
Notice that $p$ must contain vertex $v_0$ (otherwise
$p$ cannot be of length $|s|$, since $|s|=3n$). 

Now, assume that the first vertex of $p$ is $v_i$, 
with $1 \leq i \leq m$. Then, by construction, 
each position of $s$ containing $y_i$, with $1 \leq i \leq n$,
is mapped in vertex $v_0$, hence
at least $n-1$ edit operations on these positions
of $s$ 
are required for the approximate matching, plus
an edit operation either on $s$ or on the label $x_0$ of $v_0$.
But then $n > h$ edit operations are required 
by an approximate matching of $s$ with $p$. 

Assume that the first vertex of $p$ is $v_{i,j}$,
with $1 \leq i \leq m$ and $1 \leq l \leq |S_j|$.
Then the positions with symbols $y_q$ in $s$, 
with $1 \leq q \leq n$, are mapped to vertices
$v_l$, with $1 \leq l \leq m$, thus requiring
at least $n > k$ edit operations to match $s$ with $p$.

As a consequence, it follows that the first positions
of $p$ must be $v_0$ and 
hence that the positions of $s$ containing
symbol $x_0$ always match $v_0$ (no edit
operation is required). 
The positions of $s$ containing
symbol $z$ are mapped to vertices $v_i$, with $1 \leq i \leq m$. Assume that an edit operation
is done in one of those positions of $s$ and that
its symbol is edited to $x_i$. 
Since $x_i$ has only one occurrence in $s$,
each position of $s$ mapped to $v_i$ requires an edit operation.
Hence by editing the label of $v_i$ to $z$, we obtain
an approximate matching of $s$ that does not increase
the number of edit operations and such that no edit
operation is applied in the positions of $s$ mapped to $v_i$.

Finally, assume that a position $t$, with $1 \leq t \leq |s|$,
of $s$ containing symbol $y_i$, 
with $1 \leq i \leq n$, is mapped to a vertex $v_{j,l}$, 
with $1 \leq j \leq m$ and $1 \leq l \leq |S_j|$,
and that either 
$\sigma(v_{j,l})$ or the symbol $s[t]$ is edited and let
$v_a$, $1 \leq a \leq n$, be the vertex that precedes $v_{j,l}$ in $p$. 
Consider a set $S_b$, with $1 \leq b \leq m$,
that contains element $u_i$, with $1 \leq i \leq n$. 
We can compute an approximate matching of $s$ and $G$ 
that does not increase the number of edit operations, 
by replacing $v_a$ with $v_b$
and editing the label of $v_b$ to $z$ 
(if it is no already edited),
so that position $t$
of $s$ is mapped in $v_{b,q}$, with $\sigma(v_{b,q}) = y_i$,
and hence no edit operation is required  neither for symbol $\sigma(v_{b,q})$
nor for position $t$ of $s$.

It follows that, starting from $p$, 
we can compute in polynomial time a path $p'$ in $G$
that is an approximate matching of $s$, 
that does not increase the number of edit operations with 
respect to $p$ 
and such that the edit operations are applied only
to labels of the graph, and, in particular, 
only to vertices $v_i$, $1 \leq i \leq m$, in $p$. 
Then we can compute a solution of
\SetCover{} as in the proof of Lemma \ref{teo:W1Hard}:

\[
C' = \{ S_i: \text{the label of $v_i$ is edited to } z , 1 \leq i \leq m \}.
\]

As in the proof of Lemma \ref{lem:sub1}, 
at most $h$ labels of vertices of $p$ are edited (to $z$),
hence at most $h$ sets belong to $C'$.
Furthermore, since each vertex with label $y_i$,
$1 \leq i \leq n$, is connected to a vertex
$v_j$ in $p$, $1 \leq j \leq m$, by construction
it follows that 
each element of $U$ belongs to some set in $C'$. 

The parameterized reduction  from \SetCover{}
to \SeqAlign{} we have described and 
the W[2]-hardness of \SetCover{}  \cite{DBLP:journals/tcs/PazM81},
imply that \SeqAlign{} is W[2]-hard
when parameterized by the number of edit operations.

Notice that the reduction is also an
approximation preserving reduction \cite{DBLP:books/daglib/0030297},
since, given a solution of \SetCover{} consisting of $h$ sets, 
we can compute in
polynomial time a solution of \SeqAlign{} that requires $h$ edit operations and, given a solution of \SeqAlign{} that requires $h$ edit operations, we can compute in polynomial time a solution 
of \SetCover{} consisting of $h$ sets.
Since 
\SetCover{} is not approximable within factor 
$\Omega(\log (|U|))$, unless P=NP,
even when $|U|=n$ and $|C|=m$ 
are polynomially related \cite{DBLP:journals/eccc/Nelson07},
and $|V| =  m + mn + 1$, while $|s| = 3 n$,
it follows that \SeqAlign{} cannot be approximated
within factor $\Omega(\log (|V|))$ and $\Omega(\log (|s|))$, unless P=NP.
\qed
\end{proof}

%

\section{\SeqComp{} Parameterized by $|s|$}
\label{sec:paracompl1}

We present a fixed-parameter algorithm
for \SeqComp{} when parameterized by $|s|$.
We consider the case where each vertex 
of $G$ is labeled with exactly one symbol
(notice that in this case, 
by Theorem \ref{teo_hard1}, \SeqComp{} is NP-complete
and, by Corollary \ref{cor:kernel},  
\SeqComp{} parameterized by $|s|$ does not admit a polynomial
kernel unless $NP \subseteq coNP/Poly$).



We start by proving an easy property
of an instance of \SeqComp{}.

\begin{lemma}
\label{lem:tract1}
$|\Sigma| \leq |s|$.
\end{lemma}
\begin{proof}
The lemma clearly holds for those symbols having an occurrence
in $s$. For those symbols that label 
vertices of $G$ but that do not have occurrences in $s$, since they cannot exactly match any position of $s$ and
they will be edited in a restricted approximate matching of $s$,
it follows that we can substitute all their occurrences
with any symbol in $\Sigma$.
\qed
\end{proof}

The fixed-parameter algorithm is based on the
color-coding technique \cite{Alon:Yuster:Zwick:1995}
and on dynamic programming.
Consider a path $p$ in $G$ that is compatible with $s$
and the set $V(p)$ of vertices that induces $p$, where $|V(p)| = k$.
It holds $k \leq |s|$, since
each position of $s$ is mapped in at least one vertex of $p$.

We consider a coloring of $V$
with a set of colors $\{c_1 , \dots, c_k\}$, 
where, given a vertex $v \in V$, 
we denote by $c(v)$ the color
assigned to $v$.
Based on color-coding (see Definition \ref{def:perfect-hash}), we assume that
the coloring is \emph{colorful}, that is
each vertex of $V(p)$ is assigned a
distinct color in $\{c_1 , \dots, c_k\}$. 

Now, each color $c_i$, 
with $1 \leq i \leq k$, is associated by
a function 
$r $: $\{ c_1, \dots, c_k\} \rightarrow$ $\Sigma$,
with a symbol in $ \Sigma$, 
that represents the fact that
the vertices of $p$ that are colored by $c_i$, 
with $1 \leq i \leq k$, must match a position
of $s$ containing symbol $r(c_i)$. 
In this case we say that $p$ \emph{satisfies} $r$.
The algorithm iterates over the possible colorings of
graph $G$ based on a family of perfect hash functions and over the possible functions $r$.

Now, given a coloring of $G$ and a function $r$,
define a function 
$M_{r}[i,v]$, with $1 \leq i \leq |s|$ and
$v \in V$, as follows. 
$M_{r}[i,v]$ is equal to $1$ if 
there exists a path $p$
of $G$ that is compatible with $s[1,i]$ and
such that 
(1) position $i$ of $s$ is mapped in $v$,
and (2) $p$ satisfies $r$; else $M_{r}[i,v]=0$.
Notice that, since $s[1,i]$ is mapped in $v$,
it follows that $v$ is the last vertex of $p$.
Next, we describe the recurrence to compute 
$M_{r}[i,v]$. For $i \geq 2$,
if $r(c(v)) \neq s[i]$, then $M_{r}[i,v] = 0$;
if $r(c(v)) = s[i]$, then:
\[
M_{r}[i,v] = \bigvee_{u \in V:(u,v) \in E} 
M_{r}[i-1,u]
\]


In the base case, it holds $M_{r}[1,v] = 1$
if and only if $r(c(v))= s[1]$, else $M_{r}[1,v] = 0$.
Next, we prove the correctness of the recurrence.

\begin{lemma}
\label{lemma:correctness:algorithm1}
$M_{r}[i,v]$ is equal to $1$ if and only if
there exists a path $p$
of $G$ that is compatible with $s[1,i]$ and
such that 
(1) position $i$ of $s$ is mapped in $v$,
and (2) $p$ satisfies $r$.
\end{lemma}
\begin{proof}
We prove the lemma by induction on $i \geq 1$.
First, we consider the base case, when $i=1$. 
By definition it holds $M_{r}[1,v] = 1$ 
if and only if $r(c(v))= s[1]$, with $v \in V$, hence 
if and only if there exists a path consisting only of vertex 
$v$ that 
is a compatible matching of $s[1]$ and such that
$p$ satisfies $r$.

Assume that the lemma holds for $i-1 \geq 1$, 
we show that it holds for $i$.
Consider the case that $M_{r}[i,v] = 1$, it 
follows that there exists a vertex $u$, such that
$M_{r}[i-1,u] = 1$, 
$(u,v) \in E$ and $r(c(v))=s[i]$.
By induction hypothesis, there exists a path $p$ in $G$ 
that is a compatible matching of $s$
such that:
(1) the last vertex of $p$ is $u$ 
(thus position $i-1$ of $s$ is mapped in $u$),
and (2) $p$ satisfies $r$.
Then, consider the path $p'$ obtained by connecting $u$ to $v$.
Since $(u,v) \in E$ and $r(c(v))=s[i]$, it follows that $p'$ 
is a compatible matching of $s$ such that
(1) 
position $i$ of $s$ is mapped in $v$,
and (2) $p$ satisfies $r$.

Consider a path $p$ in $G$ that is a compatible matching 
of $s[1,i]$ and such that
(1) 
position $i$ of $s$ is mapped in $v$,
and (2) $p$ satisfies $r$.
Since $i-1 \geq 1$, it 
follows that there exists a vertex $u$ in $p$, with $(u,v) \in E$. Consider the path $p'$ obtained from $p$ by removing
$v$. 
By induction hypothesis, path $p'$ in $G$ is an approximate matching 
of $s[1,i-1]$ such that
(1) position $i-1$ of $s$ is mapped in $u$ (the last vertex of $p'$),
and (2) $p$ satisfies $r$.
By induction hypothesis, $M_{r}[i-1,u] = 1$.
By the definition of the recurrence, since $(u,v) \in E$ and  $r(c(v))=s[i]$,
it follows that $M_{r}[i,v] = 1$.
\qed
\end{proof}

In order to compute a colorful coloring of $G$, we 
consider a perfect family of hash functions for the set
of vertices of $G$.  


\begin{definition}
\label{def:perfect-hash}
Let $G=(V,E, \sigma)$ be a labeled graph  and let 
$C=\{c_1, \dots, c_k \}$ be a set of colors.
A family $F$ of hash functions from $V$ to $C$ is called \emph{perfect} if 
for each subset $V' \subseteq V$, with $|V'|=k$, 
there exists a function $f \in F$ such that
for each $x,y \in V'$, with $x \neq y$,
$f(x)=c_i$, $f(y)=c_j$, 
with $1 \leq i,j \leq k$ and $i \neq j$.
\end{definition}

It has been shown in~\cite{Alon:Yuster:Zwick:1995} that a perfect family $F$ of hash functions from $V$ to $C$, 
having size $ 2^{O(k)}O(\log |V| )$, can be computed
in time $2^{O(k)} O( |V| \log |V|)$. 
From Lemma \ref{lemma:correctness:algorithm1} and by
using a perfect family of hash functions to color the vertices
in $G$, we can prove the main result of this section.

\begin{theorem}
\label{teo:paracoml}
The \SeqComp{} problem can be decided in time
$2^{O(|s|)} O(|s|^{|s|+1} |V|^2 \log |V|)$.
\end{theorem}
\begin{proof}
First, we consider the correctness of the algorithm.
From Lemma \ref{lemma:correctness:algorithm1},
given a colorful coloring of $G$ and a function $r$,
there exists a path $p$ in $G$ compatible with $s$ such that the last position of $s$ is mapped
in $v$, with $v \in V$, and $p$ satisfies function $r$ 
if and only if $M_r[|s|,v] = 1$, for some
$v \in V$.

Consider a family $F$ of perfect hash functions for $G$,
by Definition \ref{def:perfect-hash}, 
for each path $p$, induced by $V(p)$, with $|V(p)|=k$,
there exists a function $f$ in $F$ that is colorful for $V(p)$,
that is it assigns to each vertex
in $V(p)$ a distinct color. 
Consider a path $p$ that is compatible with $s$, 
and for a position $t$, $1 \leq t \leq |s|$, 
of $s$ mapped to a vertex $u$ in $G$,
define function $r$ so that it assigns symbol 
$s[t]$ to color $c(u)$. 
%
Then $r$ is satisfied by $p$.
The correctness of the algorithm holds since
it iterates over the possible functions
in a family of perfect hash functions and over 
the possible definitions of $r$.

It has been shown in~\cite{Alon:Yuster:Zwick:1995} that a perfect family $F$ of hash functions from $V$ to $C$, where $F$ has size 
$2^{O(k)} O(\log |V|)$, can be computed
in time $2^{O(k)} O( |V| \log |V|)$. 

The function $r$ can be computed in time $O(|s|^{|s|})$,
since $k \leq |s|$, each color can be associated
with at most $\Sigma$ symbols and, 
by Lemma \ref{lem:tract1}, $|\Sigma | \leq |s|$.
The dynamic programming
recurrence requires time $O(|s| |V|^2)$, since
$M_r[i,v]$ consists of $|s||V|$ entries and each
entry is computed in $O(|V|)$ time.
Hence the overall time complexity is 
$2^{O(k)} O(\log |V|) O(|s|^{|s|} (s|V|^2))$, and,
since $k \leq |s|$,
we can conclude that the time complexity is $2^{O(|s|)} O(|s|^{|s|+1} |V|^2 \log |V|)$.
\qed
\end{proof}


\section{Conclusion}
\label{sec:conclusion}

In this contribution we have presented results 
on the tractability of 
the approximate matching of a
query string to a labeled graph.
There are several open questions related
to variants of this problem.
It
will be interesting to further investigate
the approximability of \SeqAlign{}, 
since it can be trivially
approximated within factor $|s|$ in polynomial time, 
while it cannot be approximated within factor 
$\Omega (\log(|s|))$, unless P=NP.
Another interesting open question is to 
investigate the parameterized complexity of \SeqAlign{} when the edit operations
are not restricted to symbol substitutions, 
but include symbol insertions and 
deletions.

\bibliographystyle{splncs04}
\bibliography{biblio}

\begin{thebibliography}{10}
\providecommand{\url}[1]{\texttt{#1}}
\providecommand{\urlprefix}{URL }
\providecommand{\doi}[1]{https://doi.org/#1}

\bibitem{DBLP:conf/cpm/Akutsu93}
Akutsu, T.: A linear time pattern matching algorithm between a string and a
  tree. In: Combinatorial Pattern Matching, 4th Annual Symposium, {CPM} 93,
  Padova, Italy, June 2-4, 1993, Proceedings. pp. 1--10 (1993)

\bibitem{Alon:Yuster:Zwick:1995}
Alon, N., Yuster, R., Zwick, U.: Color-coding. Journal of the ACM
  \textbf{42(4)},  844--856 (1995)

\bibitem{DBLP:journals/jal/AmirLL00}
Amir, A., Lewenstein, M., Lewenstein, N.: Pattern matching in hypertext. J.
  Algorithms  \textbf{35}(1),  82--99 (2000)

\bibitem{DBLP:journals/jcss/BodlaenderDFH09}
Bodlaender, H.L., Downey, R.G., Fellows, M.R., Hermelin, D.: On problems
  without polynomial kernels. J. Comput. Syst. Sci.  \textbf{75}(8),  423--434
  (2009)

\bibitem{DBLP:journals/tcs/BodlaenderJK13}
Bodlaender, H.L., Jansen, B.M.P., Kratsch, S.: Kernel bounds for path and cycle
  problems. Theor. Comput. Sci.  \textbf{511},  117--136 (2013)

\bibitem{DBLP:series/txcs/DowneyF13}
Downey, R.G., Fellows, M.R.: Fundamentals of Parameterized Complexity. Texts in
  Computer Science, Springer (2013)

\bibitem{DBLP:conf/icalp/EquiGMT19}
Equi, M., Grossi, R., M{\"{a}}kinen, V., Tomescu, A.I.: On the complexity of
  string matching for graphs. In: Baier, C., Chatzigiannakis, I., Flocchini,
  P., Leonardi, S. (eds.) 46th International Colloquium on Automata, Languages,
  and Programming, {ICALP} 2019, July 9-12, 2019, Patras, Greece. LIPIcs,
  vol.~132, pp. 55:1--55:15. Schloss Dagstuhl - Leibniz-Zentrum fuer Informatik
  (2019)

\bibitem{garey}
Garey, M.R., Johnson, D.S.: {Computers and Intractability: A Guide to the
  Theory of NP-Completeness}. WH Freeman \& Co. (1979)

\bibitem{DBLP:conf/recomb/JainZGA19}
Jain, C., Zhang, H., Gao, Y., Aluru, S.: On the complexity of sequence to graph
  alignment. In: Cowen, L.J. (ed.) Research in Computational Molecular Biology
  - 23rd Annual International Conference, {RECOMB} 2019, Washington, DC, USA,
  May 5-8, 2019, Proceedings. Lecture Notes in Computer Science, vol. 11467,
  pp. 85--100. Springer (2019)

\bibitem{Manber1992}
Manber, U., Wu, S.: Approximate string matching with arbitrary cost for text
  and hypertext. In: Advances in Structural and Syntactic Pattern Recognition.
  pp. 22--33 (1992)

\bibitem{DBLP:conf/eccb/Myers05}
Myers, E.W.: The fragment assembly string graph. Bioinformatics  \textbf{21
  (suppl\_2)},  ii79--ii85 (2005)

\bibitem{DBLP:journals/tcs/Navarro00}
Navarro, G.: Improved approximate pattern matching on hypertext. Theor. Comput.
  Sci.  \textbf{237}(1-2),  455--463 (2000)

\bibitem{DBLP:journals/eccc/Nelson07}
Nelson, J.: A note on set cover inapproximability independent of universe size.
  Electronic Colloquium on Computational Complexity {(ECCC)}  \textbf{14}(105)
  (2007)

\bibitem{DBLP:journals/jcb/NguyenHZREAKHP15}
Nguyen, N., Hickey, G., Zerbino, D.R., Raney, B.J., Earl, D., Armstrong, J.,
  Kent, W.J., Haussler, D., Paten, B.: Building a pan-genome reference for a
  population. Journal of Computational Biology  \textbf{22}(5),  387--401
  (2015)

\bibitem{Niedermeier:2006}
Niedermeier, R.: Invitation to Fixed-Parameter Algorithms. Oxford University
  Press (2006)

\bibitem{DBLP:conf/cpm/ParkK95}
Park, K., Kim, D.K.: String matching in hypertext. In: Galil, Z., Ukkonen, E.
  (eds.) Combinatorial Pattern Matching, 6th Annual Symposium, {CPM} 95, Espoo,
  Finland, July 5-7, 1995, Proceedings. Lecture Notes in Computer Science,
  vol.~937, pp. 318--329. Springer (1995)

\bibitem{DBLP:journals/tcs/PazM81}
Paz, A., Moran, S.: Non deterministic polynomial optimization problems and
  their approximations. Theoretical Computer Science  \textbf{15},  251--277
  (1981)

\bibitem{Pevzner2001}
Pevzner, P., Tang, H., Waterman, M.S.: An eulerian path approach to dna
  fragment assembly. Proceedings of the National Academy of Sciences
  \textbf{98}(17),  9748--97533 (2001)

\bibitem{pangenomics}
{The Computational Pan{-}Genomics Consortium}: Computational pan-genomics:
  status, promises and challenges. Briefings in Bioinformatics  \textbf{19}(1),
   118--135 (2018)

\bibitem{DBLP:books/daglib/0030297}
Williamson, D.P., Shmoys, D.B.: The Design of Approximation Algorithms.
  Cambridge University Press (2011)

\end{thebibliography}

\newpage

\end{document}